\documentclass[11pt,a4paper]{article}
\usepackage[utf8]{inputenc}
\usepackage[T1]{fontenc}
\usepackage{amsmath,amssymb,amsthm,bm,mathtools}
\usepackage{geometry}
\usepackage{hyperref}
\usepackage{enumitem}
\usepackage{booktabs}
\usepackage{graphicx}
\DeclareGraphicsExtensions{.pdf}
\usepackage{microtype}
\usepackage{parskip}
\hypersetup{colorlinks=true,linkcolor=blue,citecolor=blue,urlcolor=blue,
  pdfkeywords={lattice Boltzmann method, multiple-relaxation-time, CPTP maps, open quantum systems, two-rail encoding}}

\newtheorem{theorem}{Theorem}[section]
\newtheorem{proposition}[theorem]{Proposition}
\newtheorem{corollary}[theorem]{Corollary}
\newtheorem{remark}[theorem]{Remark}
\newtheorem{assumption}[theorem]{Assumption}

\newcommand{\R}{\mathbb{R}}
\newcommand{\Tr}{\operatorname{Tr}}
\newcommand{\eqdef}{\stackrel{\rm def}{=}}
\newcommand{\AD}{\mathcal{A}}
\newcommand{\Ec}{\mathcal{E}}
\newcommand{\Dc}{\mathcal{D}}
\let\Dd\Dc
\newcommand{\Cc}{\mathcal{C}}

\title{Deterministic Realization of Classical Dissipation on Quantum Computers}

\author{%
  \begin{tabular}{@{}c@{\hspace{1.75em}}c@{\hspace{1.75em}}c@{}}
    Muhammad Idrees Khan$^{a,b}$\protect\thanks{Corresponding author.} & Sauro Succi$^{c,d}$ & Hua-Dong Yao$^{a}$
  \end{tabular}\\[0.75em]
  \makebox[\textwidth][c]{%
    \begin{minipage}{\textwidth}
      \centering
      $^{a}$\,Department of Mechanical Engineering, Chalmers University of Technology, Gothenburg, 41296, Sweden\\[0.45em]
      $^{b}$\,Department of Enterprise Engineering ``Mario Lucertini'', University of Rome ``Tor Vergata'', Via del Politecnico 1, 00133 Rome, Italy\\[0.45em]
      $^{c}$\,Italian Institute of Technology, Viale Regina Elena, 291, 00161 Roma, Italy\\[0.45em]
      $^{d}$\,Department of Physics, Harvard University, 33 Oxford Street, Cambridge, 02138, Massachusetts, USA
    \end{minipage}%
  }%
}
\date{}

\begin{document}
\maketitle

\begin{abstract}
Lattice Boltzmann (LB) on quantum devices must reconcile unitary gate evolution with the dissipative \emph{collision} step. In the multiple-relaxation-time (MRT) class, we work in the common setting of \emph{modewise diagonal} moment relaxation, $\delta m_r'=\lambda_r\,\delta m_r$ with $\lambda_r\in[-1,1]$ (overrelaxation if $\lambda_r<0$). Embedding that contraction in a unitary by block encoding or a linear combination of unitaries (LCU) typically yields subunitary success probability that decays multiplicatively across modes, sites, and time, a key bottleneck for quantum LB. \emph{For the dissipative MRT block alone} we give a \emph{block-encoding-free} construction: a signed \emph{two-rail} population encoding, then a completely positive trace-preserving (CPTP) map (per-rail amplitude damping with survival $|\lambda_r|$ and, if $\lambda_r<0$, a rail SWAP) so that, after the decode, the map agrees with classical MRT relaxation exactly (expectations of the rail number operators, common encoding--decode scale). Trace preservation gives success probability $1$ for that substage. The main result is the dissipative MRT block; construction of the equilibrium moment vector~$m^{\mathrm{eq}}=Mf^{\mathrm{eq}}$ (prescribed~$f^{\mathrm{eq}}$, host moment matrix~$M$; notation as in Section~\ref{subsec:generic-mrt}), moment transforms, streaming, and boundaries are composed with it as in a standard host pipeline and lie outside the scope of the formal theorem. Hybrid and fully coherent encodings, adaptive scales, Carleman-based context, and a one-rail no-go in the same nonnegative population framework are in the main text. Audits of the open-channel map on a long LBM collide–stream simulation and on stencil-free inputs both match the target to machine precision.
\end{abstract}


\section{Introduction}

\noindent Recently, major attention has been directed to the development of quantum algorithms for classical physics problems, most notably nonlinear transport and fluid phenomena~\cite{succi2023epl,wang2025npj,tiwari2025algorithmic}.
Such algorithms face two major extra-challenges: nonlinearity and non-unitarity (dissipation).
Several strategies have been developed in the last few years, a particularly prominent being based on Carleman linearization of the lattice Boltzmann (LB) formulation~\cite{higuera1989cylinder,higuera1989enhanced,higuera1989epl} of fluid dynamics, coupled with linear combination of unitaries (LCU) and/or block encoding (BE)~\cite{sanavio2024lbq,liu2026qlbmlinear}.
Despite numerous attempts, these algorithms are still plagued by a severe bottleneck due to the low success probability of the dissipative update.
In this paper we show that such bottleneck can be removed by turning to a deterministic completely positive trace-preserving (CPTP) channel: amplitude damping, which completely dispenses with LCU and BE steps, thereby delivering unit success probability.
The strategy is custom-made on a pseudo-spectral version of the LB method, whereby free-streaming occurs in real space, whereas the multiple-relaxation-time (MRT) collision, responsible for loss of unitarity, is performed in the space of kinetic moments.
This way, both steps rely upon diagonal representations.
This makes it possible to import powerful techniques from quantum information science and the theory of open quantum dissipative systems and to adapt them to the case of nonlinear dissipative classical systems.
Choi \emph{et al.}~\cite{choi2026lindbladian} combine homotopy analysis with Lindbladian dynamics for nonlinear partial differential equations (PDEs) in an open-system framework. This work gives a different contribution: Theorem~\ref{thm:main} provides a deterministic, block-encoding-free CPTP realization of the modewise MRT update $\delta m_r'=\lambda_r\,\delta m_r$ on signed two-rail encodings in kinetic-moment space.

In the MRT formulation of the lattice Boltzmann method (LBM)~\cite{benzi1992lbe,dhumieres2002multiple,lallemand2003theory}, the dissipative part of the collision is modewise diagonal in a chosen moment basis: nonequilibrium moments relax as $\delta m_r'=\lambda_r\,\delta m_r$ with $\lambda_r=1-s_r\in[-1,1]$. For every dissipative mode this is a \emph{contraction} ($|\lambda_r|\le 1$), and in the overrelaxation regime $\lambda_r<0$ it is also sign-reversing. Contractive, non-unitary maps are not directly implementable on a quantum computer: the standard route is to embed the contraction inside a larger unitary (block encoding) and post-select on an ancilla flag~\cite{nielsen2010qcqi,watrous2018quantum}. The per-step success probability of such a block-encoded contraction is bounded above by $|\lambda_r|^2/\alpha_r^2$ for subnormalization $\alpha_r\ge 1$ (equivalently by $|\lambda_r|^2$ in the tight case $\alpha_r=1$), and compounds \emph{multiplicatively} over dissipative modes, lattice sites, and timesteps. Multiplicative post-selection or heralding penalties of the same structural form also appear in closely related quantum lattice Boltzmann method (QLBM) routes, notably in LCU-based realizations of linear collision operators~\cite{liu2026qlbmlinear} and in block-encoded Carleman-linearized lattice Boltzmann constructions~\cite{sanavio2024lbq}; this compounded probability is widely identified in the QLBM literature as the principal obstacle to end-to-end execution at useful Reynolds numbers.

Theorem~\ref{thm:main} makes the CPTP construction precise: after a signed two-rail population encoding of $\delta m_r$ (Section~\ref{sec:encoding}), amplitude damping with survival $\alpha_r=|\lambda_r|$ on each rail and a rail SWAP whenever $\lambda_r<0$ reproduce $\delta m_r\mapsto \lambda_r\,\delta m_r$ exactly at the decoded expectation-value level, with per-step success probability $1$ for the dissipative block itself (Corollary~\ref{cor:deterministic}).

The remainder of the paper is organized as follows. Section~\ref{subsec:generic-mrt} fixes notation for diagonal MRT in moment space. Section~\ref{sec:encoding} contains the encoding, the signed channel, and the main theorem. Section~\ref{sec:blockenc-vs-cptp} quantifies the block-encoding / success-probability comparison and situates the present result in the QLBM literature. Section~\ref{sec:integration} sketches how the primitive integrates into a generic block-encoded MRT collision pipeline, comments on its relation to Carleman-linearized QLBM circuits, which share the same structural bottleneck but do not expose a mode-wise diagonal target, and spells out a minimal hybrid Carleman--CPTP pipeline (\S\ref{subsec:carleman-cptp}) in which Carleman linearization is used only for the nonlinear equilibrium evaluation and the dissipative relaxation is handled by the present primitive. Section~\ref{sec:validation} reports a machine-precision numerical audit on D3Q19.

The main contributions are the following.
\begin{enumerate}[leftmargin=2em,label=(C\arabic*)]
\item The dissipative MRT update $\delta m_r'=\lambda_r\delta m_r$ is realized as a \emph{deterministic CPTP channel with unit per-step success probability}, in contrast to block-encoded contractions whose success probability decays multiplicatively across modes, sites, and timesteps.
\item Theorem~\ref{thm:main} proves exact decoded-expectation equivalence on two-rail registers, including the overrelaxation branch $\lambda_r<0$, which is captured by a classically conditioned rail SWAP composed with amplitude damping at survival $|\lambda_r|$.
\item Proposition~\ref{prop:adaptive} shows that the encoding scale $S_r(\bm x,t)$ is classically indexed side information rather than a hidden quantum resource: adaptive scales produce a classically indexed family of encode--channel--decode composites.
\item Within the single-rail nonnegative population-encoding framework of this paper, one rail cannot exactly represent a signed nonequilibrium moment on a two-signed domain (Appendix~\ref{app:one-rail}), which motivates the two-rail split.
\item Numerical audits: Section~\ref{subsec:d3q19-audit} (D3Q19, decaying Taylor--Green vortex) confirms machine-precision agreement with the target dissipative update when the error is maximized over all dissipative modes and all lattice sites; Section~\ref{subsec:synthetic-audit} independently stress-tests the same per-mode map on stencil-free random $(\delta m_r,\lambda_r)\in[-X,X]\times[-1,1]$, with no particular stencil or flow.
\end{enumerate}

Any QLBM pipeline whose dissipative collision is implemented by block-encoding a diagonal contraction on nonequilibrium moments (that is, any block-encoded MRT scheme in which the dissipative block reduces to $\Lambda=\mathrm{diag}(\lambda_r)$ on dissipative modes) pays a per-mode, per-site, per-step success-probability cost that decays multiplicatively in $T$, $|\Omega_h|$, and $|\Dd|$ (Section~\ref{subsec:be-bottleneck}). Substituting the primitive proved here for that contraction replaces this compounded cost by $1$ for the dissipative block itself, independently of $\lambda_r$, $S_r$, and the input state (Corollary~\ref{cor:deterministic}); the overrelaxation branch $\lambda_r<0$ is absorbed as a classically conditioned rail SWAP at no probability cost, not as an additional block-encoded conditional. The price paid in exchange, namely a classical nonlinear sign split, two rails per dissipative mode, and classically tracked scales $S_r(\bm x,t)$, is bookkeeping rather than probability loss (Section~\ref{subsec:tradeoffs}). Preparation, extraction, reconstruction, streaming, boundaries, and readout remain separate, host-determined pipeline stages (Section~\ref{sec:integration}) and are not claimed to benefit from Theorem~\ref{thm:main} in isolation. Block-encoded Carleman-linearized lattice Boltzmann circuits~\cite{sanavio2024lbq} share the same structural bottleneck of a non-unitary collision operator embedded with ancilla post-selection, but their non-unitarity is carried by the full (non-sparse, non-block-diagonal) Carleman collision matrix rather than by a mode-wise diagonal $\Lambda$, so the present primitive does not apply to them as a drop-in replacement; the relation is discussed in Section~\ref{subsec:dropin}.

The focus of this article is the dissipative block of MRT collision and its open-channel realization on rails. Coherent equilibrium loading, full one-node collision reconstruction, streaming, lattice boundaries, hardware execution, finite-shot behavior, and asymptotic speedup claims are outside the scope of this work. The dissipative primitive proved here is of standalone theoretical interest as a block-encoding-free CPTP realization of a signed modewise contraction: its validity does not depend on any specific host pipeline, and Section~\ref{sec:integration} shows that it is compatible with generic block-encoded MRT collision pipelines. Its relation to Carleman-linearized QLBM circuits, which illustrate the same block-encoding bottleneck but do not expose a mode-wise diagonal target, is also discussed there.

\noindent Throughout the paper, items are labeled (Letter)(number) so they can be cited in one pair of parentheses. The letter only names which list is meant (it is not a variable): C marks the contribution list; I marks the practical-interpretation checklist; T marks the three tradeoffs of the deterministic CPTP route; P marks the four on-site MRT collision pipeline stages; U marks the three host specializations in the list immediately after~\eqref{eq:high-level-hybrid} (block-encoding, LCU, or Carleman); H marks the six substeps in the hybrid Carleman--CPTP write-up~\eqref{eq:hybrid-stages}; and S marks the five stencil-free synthetic audit cases. The number indexes an entry within that list. Thus (C1)--(C5) are the items above; (I1)--(I7) appear in Section~\ref{subsec:practical}; (T1)--(T3) in Section~\ref{subsec:tradeoffs}; (P1)--(P4) in Section~\ref{subsec:generic-pipeline}; (U1)--(U3) in that list (right after~\eqref{eq:high-level-hybrid}); (H1)--(H6) in~\eqref{eq:hybrid-stages}; and (S1)--(S5) in Section~\ref{subsec:synthetic-audit}.

\section{MRT-LBM in moment space}

\subsection{Generic moment-space collision}\label{subsec:generic-mrt}
Consider a lattice Boltzmann scheme with $q$ discrete velocities $\{\bm c_i\}_{i=0}^{q-1}$ and population vector
\begin{equation}
 f=(f_0,\ldots,f_{q-1})^\top\in\R^{q}.
\end{equation}
Let $M\in\R^{q\times q}$ be any fixed \emph{invertible} moment matrix and define $m=Mf$. Given a prescribed equilibrium $f^{\mathrm{eq}}(\rho,\bm u)$, where $(\rho,\bm u)$ are the hydrodynamic fields extracted by the baseline MRT model (the usual mass and momentum densities constructed from $f$), set $m^{\mathrm{eq}}=Mf^{\mathrm{eq}}$. The MRT collision step is
\begin{equation}
 \delta m = m - m^{\mathrm{eq}},
\end{equation}
\begin{equation}
 \delta m_r' = (1-s_r)\,\delta m_r,
 \qquad
 m^+ = m^{\mathrm{eq}} + \delta m',
 \qquad
 f^\star = M^{-1}m^+.
 \label{eq:mrt-collision}
\end{equation}
Post-collision populations stream by
\begin{equation}
 f_i(\bm x+\bm c_i,t+1)=f_i^\star(\bm x,t).
 \label{eq:streaming}
\end{equation}
Let $\Cc$ denote the conserved moment index set and $\Dd$ its complement (dissipative modes), with $s_r=0$ for $r\in\Cc$. Define the modewise relaxation multiplier
\begin{equation}
 \lambda_r \eqdef 1-s_r.
 \label{eq:lambda-def}
\end{equation}
Under the usual MRT admissible range $0\le s_r\le 2$ for $r\in\Dd$, one has
\begin{equation}
 \lambda_r\in[-1,1]
 \qquad \text{for every } r\in\Dd.
 \label{eq:lambda-range}
\end{equation}
The regime $\lambda_r<0$ is \emph{overrelaxation}. Below we only require $M$ to be invertible; orthogonality is not assumed. The collision operator in \eqref{eq:mrt-collision} is \emph{modewise diagonal} in the moment coordinates defined by $M$: each nonequilibrium moment $\delta m_r$ relaxes with its own scalar factor. Block-coupled or collision models that are nondiagonal in this basis lie outside the scope of the theorem below. The open-channel construction in Section~\ref{sec:encoding} applies independently to each $r\in\Dd$ whenever \eqref{eq:mrt-collision} holds in that basis; particular stencils (D2Q9, D3Q19, D3Q27, \ldots) specify $q$, the $\bm c_i$, weights, and the numerical values of $s_r$.

\section{Signed two-rail population encoding and open-channel realization}\label{sec:encoding}
Fix the generic MRT setting of Section~\ref{subsec:generic-mrt}. The central result is Theorem~\ref{thm:main}, which addresses dissipative modes only; Appendix~\ref{app:classical-embed} records optional algebraic consequences when the open dissipative update is embedded in the same moment reconstruction as a classical baseline. D3Q19 (Appendix~\ref{app:d3q19}) appears only in the numerical consistency check of Section~\ref{sec:validation}.

For each dissipative mode $r\in\Dd$, let $S_r=S_r(\bm x,t)>0$ denote the \textbf{encoding scale}, which may be chosen globally, per mode, or adaptively per node and time, provided the normalized rail populations defined below remain in $[0,1]$. Define
\begin{equation}
 p_r^+ = \frac{\max(\delta m_r,0)}{S_r},
 \qquad
 p_r^- = \frac{\max(-\delta m_r,0)}{S_r}.
 \label{eq:rails}
\end{equation}
We require $p_r^\pm\in[0,1]$; the scale $S_r$ carries the physical magnitude and is restored by the decoder. Assign one qubit to each rail and define diagonal one-qubit states
\begin{equation}
 \rho_r^+ = (1-p_r^+)\,|0\rangle\langle 0| + p_r^+\,|1\rangle\langle 1|,
 \qquad
 \rho_r^- = (1-p_r^-)\,|0\rangle\langle 0| + p_r^-\,|1\rangle\langle 1|.
\end{equation}
The signed two-rail population encoding is $\rho_r = \rho_r^+\otimes \rho_r^-$. Introduce rail number operators $n_+ = |1\rangle\langle 1|\otimes I$ and $n_- = I\otimes |1\rangle\langle 1|$. The decoder is
\begin{equation}
 \Dc_r(\rho_r) \eqdef S_r\big(\Tr[n_+\rho_r]-\Tr[n_-\rho_r]\big)
 = S_r(p_r^+-p_r^-)=\delta m_r.
 \label{eq:decoder}
\end{equation}

\begin{remark}[Choosing the encoding scale~$S_r$]\label{rem:choose-Sr}
At each use, feasibility $p_r^\pm\in[0,1]$ in \eqref{eq:rails} is equivalent to $S_r\ge |\delta m_r|$ when $\delta m_r\neq 0$ (if $\delta m_r=0$, any $S_r>0$ suffices). An automatic per-step choice at fixed $(\bm x,t)$ is $S_r=\max(|\delta m_r|,\varepsilon)$ with fixed $\varepsilon>0$ for numerical stability, using the \emph{same} $S_r$ in encoding and decoding \eqref{eq:decoder}. A manual or campaign-wide choice is any fixed $S_r$ (or per-mode bound) with $S_r\ge \sup|\delta m_r|$ over the nodes, times, and runs of interest; larger $S_r$ remains valid but dilutes rail populations. The theorem does not fix $S_r$ from Reynolds number or mesh spacing alone without such state information or a chosen bound.
\end{remark}

\begin{remark}[Nonlinear encoding map]\label{rem:nonlinear-encoding}
For fixed $S_r$, the map $\delta m_r\mapsto(p_r^+,p_r^-)$ defined by \eqref{eq:rails} is \emph{nonlinear} because of the positive/negative split. Accordingly, the result is an exact decoded CPTP realization on encoded rail registers (nonlinear encoding, CPTP channels on the rails, then decoding), \emph{not} a single linear CPTP map acting directly on the unencoded classical moment variable~$\delta m_r$.
\end{remark}

The positive/negative split in \eqref{eq:rails} is not a technical convenience: within the present nonnegative population-encoding framework, a single rail cannot exactly represent a signed scalar on any domain containing both signs. We defer the precise statement and proof to Appendix~\ref{app:one-rail}, since the construction below does not rely on it; it is recorded there to motivate the two-rail split rather than a single-rail encoding.

Let $\AD_\alpha$ denote the one-qubit amplitude-damping channel with survival factor $\alpha\in[0,1]$, which maps population $p\mapsto \alpha p$ \cite{nielsen2010qcqi,watrous2018quantum}. Set $\alpha_r=|\lambda_r|$. Define the signed channel on the two-rail register by
\begin{equation}
 \Ec_r =
 \begin{cases}
 \AD_{\alpha_r}\otimes \AD_{\alpha_r}, & \lambda_r\ge 0,\\[1mm]
 \mathrm{SWAP}\circ (\AD_{\alpha_r}\otimes \AD_{\alpha_r}), & \lambda_r<0.
 \end{cases}
 \label{eq:signed-channel}
\end{equation}
Both branches are CPTP. The full construction is summarized in Figure~\ref{fig:schematic}.

\begin{figure}[htbp]
\centering
\includegraphics[width=\textwidth]{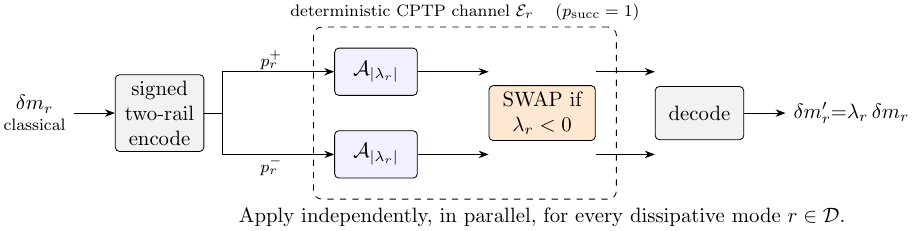}
\caption{Block-encoding-free CPTP realization of the diagonal MRT dissipative update. The classical nonequilibrium moment $\delta m_r$ is mapped onto two rails $(p_r^+,p_r^-)$ by the signed population encoding \eqref{eq:rails}. Each rail is damped by an amplitude-damping channel with survival $|\lambda_r|$, and a rail SWAP is applied when $\lambda_r<0$; the decoder \eqref{eq:decoder} reads off $\delta m_r'=\lambda_r\,\delta m_r$ exactly (Theorem~\ref{thm:main}). The boxed region is trace-preserving by construction, so the per-application success probability is unity: no heralding, post-selection, or amplitude amplification is required, in contrast to block-encoded realizations of the same contraction (Section~\ref{sec:blockenc-vs-cptp}). The SWAP branch is selected by the classical sign of $\lambda_r$ and carries no probability cost.}
\label{fig:schematic}
\end{figure}

\begin{figure}[htbp]
\centering
\includegraphics[width=\textwidth]{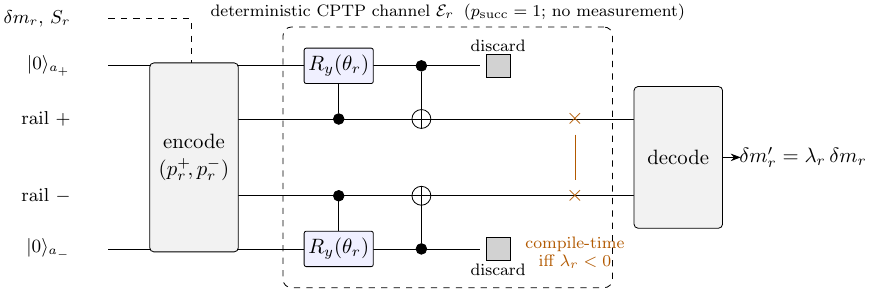}
\caption{Gate-level realization of the CPTP channel $\Ec_r$ of Figure~\ref{fig:schematic}. Each amplitude-damping box is opened into its Stinespring dilation: a fresh ancilla $|0\rangle_{a_\pm}$, a rail-controlled rotation $R_y(\theta_r)$ with $\theta_r=2\arccos\!\sqrt{|\lambda_r|}$ (equivalently $2\arcsin\!\sqrt{1-|\lambda_r|}$), a controlled NOT (CNOT) gate from the ancilla back onto the rail, and a final \emph{discard} (trace-out) of the ancilla. The rail SWAP is a \emph{compile-time} branch selected by $\mathrm{sign}(\lambda_r)$ per \eqref{eq:signed-channel}, either present (if $\lambda_r<0$) or absent, so no classical control wire or runtime conditional is needed (no coherent conditional on $\lambda_r$ is required; Section~\ref{subsec:generic-pipeline}). Crucially, no ancilla is ever measured: the dashed CPTP frame contains only coherent unitaries followed by an ancilla trace-out, both of which are deterministic, so $p_{\mathrm{succ}}=1$ per application (Corollary~\ref{cor:deterministic}). This is the concrete gate-level contrast with measurement-based realizations of the same contraction, where ancilla post-selection introduces a $\|\cdot\|$-dependent success-probability penalty (Section~\ref{sec:blockenc-vs-cptp}).}
\label{fig:circuit-h4}
\end{figure}
\begin{theorem}[Modewise CPTP realization of MRT dissipative relaxation]\label{thm:main}
For each dissipative mode $r\in\Dd$, let $S_r=S_r(\bm x,t)$ be any positive scale such that the encoded populations $p_r^\pm$ in \eqref{eq:rails} lie in $[0,1]$, and use the same $S_r$ consistently in encoding and decoding. The map $\delta m_r\mapsto(p_r^+,p_r^-)$ is generally nonlinear (Remark~\ref{rem:nonlinear-encoding}); the CPTP map $\Ec_r$ acts on the two-qubit rail state after this encoding. Under amplitude damping with $\alpha_r=|\lambda_r|$, followed by SWAP when $\lambda_r<0$, the decoded output satisfies
\begin{equation}
 \delta m_r' = \lambda_r\,\delta m_r
\end{equation}
exactly at the decoded expectation-value level. Therefore the dissipative part of the classical MRT collision operator is reproduced exactly by the encode--channel--decode construction on the rails (for any fixed valid scale choice); when the scale is chosen adaptively, Proposition~\ref{prop:adaptive} describes the resulting indexed family of composites.
\end{theorem}

\begin{proof}
Amplitude damping scales $p_r^\pm\mapsto \alpha_r p_r^\pm$. If $\lambda_r\ge 0$: $\alpha_r=\lambda_r$, no swap. Then $\Dc_r(\Ec_r(\rho_r)) = S_r(\alpha_r p_r^+ - \alpha_r p_r^-) = \lambda_r S_r(p_r^+-p_r^-) = \lambda_r\delta m_r$. If $\lambda_r<0$: $\alpha_r=-\lambda_r$, and the swap exchanges the rails after damping, so $\Dc_r(\Ec_r(\rho_r)) = S_r(\alpha_r p_r^- - \alpha_r p_r^+) = -\alpha_r S_r(p_r^+-p_r^-) = \lambda_r\delta m_r$.
\end{proof}

\begin{corollary}[Channel stage is deterministic]\label{cor:deterministic}
For every dissipative mode $r\in\Dd$ and every admissible two-rail input state $\rho_r$, the channel $\Ec_r$ of \eqref{eq:signed-channel} satisfies
\begin{equation}
\Tr\!\big[\Ec_r(\rho_r)\big] \;=\; \Tr[\rho_r] \;=\; 1,
\qquad
p_{\mathrm{succ}}^{(r)} \;=\; 1,
\label{eq:cor-det}
\end{equation}
independently of $\lambda_r$, $S_r$, and the input state. Composed independently across $r\in\Dd$, lattice sites $\bm x\in\Omega_h$, and timesteps $t=1,\ldots,T$, the product dissipative channel
\begin{equation}
\bigotimes_{t=1}^{T}\bigotimes_{\bm x\in\Omega_h}\bigotimes_{r\in\Dd}\Ec_r(\bm x,t)
\label{eq:cor-product}
\end{equation}
is CPTP with per-step and end-to-end success probability equal to $1$. In particular, the dissipative block itself does not accumulate any multiplicative success-probability penalty in $|\Dd|$, $|\Omega_h|$, or $T$, in contrast to block-encoded realizations of the same contraction (Section~\ref{subsec:be-bottleneck}).
\end{corollary}

\begin{proof}
Amplitude damping $\AD_{\alpha_r}$ with $\alpha_r\in[0,1]$ is CPTP \cite{nielsen2010qcqi,watrous2018quantum}. Tensor products of CPTP maps are CPTP, and composition with the unitary SWAP (a CPTP map) preserves CPTP. Both branches of \eqref{eq:signed-channel} are therefore CPTP, giving \eqref{eq:cor-det} by trace preservation. The product channel \eqref{eq:cor-product} is a tensor product of CPTP maps and is itself CPTP, hence trace-preserving, which gives unit end-to-end success probability.
\end{proof}

\begin{remark}[Scope of Corollary~\ref{cor:deterministic}]\label{rem:cor-scope}
The statement concerns only the dissipative CPTP block $\bigotimes_r\Ec_r$. Preparation of the rails from classical side information or from a coherent population register, equilibrium loading, nonequilibrium extraction, reconstruction $f^\star=M^{-1}m^+$, streaming, boundary treatment, and finite-shot readout are separate pipeline stages (see Section~\ref{sec:integration}); their probabilistic costs, if any, are not covered by Corollary~\ref{cor:deterministic}.
\end{remark}

Because admissible scale choices $S_r(\bm x,t)$ may vary across nodes, times, and modes, it is useful to state explicitly what remains fixed and what is classically indexed in the construction.

\begin{proposition}[Adaptive scales: classically indexed encode--channel--decode family]\label{prop:adaptive}
Suppose $S_r=S_r(\bm x,t)$ may vary with node and time (subject to $p_r^\pm\in[0,1]$). For fixed relaxation physics, the middle map $\Ec_r$ in \eqref{eq:signed-channel} depends only on $\lambda_r$ (hence on the chosen $s_r$), not on $S_r$. What adapts with $S_r(\bm x,t)$ is the \emph{composite} update $\Dc_r\circ\Ec_r\circ(\text{encode})$ on dissipative mode~$r$. Thus, as $(\bm x,t)$ and $r$ vary, one obtains a \emph{nodewise}, classically indexed family of encode--channel--decode maps; equivalently, classical side information $(\bm x,t)$ (and mode label~$r$) indexes which composite acts.
It is not, in general, a single fixed CPTP map acting directly on the raw (unencoded) nonequilibrium moments~$\delta m_r$ (cf.\ Remark~\ref{rem:nonlinear-encoding}).
\end{proposition}
\begin{proof}
Fix $(\bm x,t)$ and $r\in\Dd$. Encoding with the scale $S_r(\bm x,t)$ produces a product state $\rho_r$ on the rails; $\Ec_r$ in \eqref{eq:signed-channel} is CPTP and independent of $S_r$. Letting $(\bm x,t)$ vary replaces one such triple $(\text{encode},\Ec_r,\text{decode})$ by another only through the classically chosen~$S_r$. The composite map on unencoded $\delta m_r$ alone incorporates this $S_r$-dependent encoding and need not coincide with any single CPTP channel on the moment variable (Remark~\ref{rem:nonlinear-encoding}).
\end{proof}

\begin{remark}[Node-level composite dissipative update]\label{rem:composite}
With disjoint two-qubit registers for distinct $r\in\Dd$, applying $\Ec_r$ mode-wise and decoding with the same $S_r$ used in encoding yields $\delta m_r'=\lambda_r\delta m_r$ simultaneously for all $r\in\Dd$, i.e.\ the full dissipative moment vector update matches the classical MRT operator on the dissipative subspace. This follows from Theorem~\ref{thm:main} and independence of registers.
\end{remark}

\begin{remark}[Stencil and modewise structure]
Theorem~\ref{thm:main} assumes only modewise diagonal relaxation \eqref{eq:mrt-collision} with $\lambda_r\in[-1,1]$ and invertible $M$; block-coupled collisions in the same moment basis are out of scope.
\end{remark}

\begin{remark}[Encoding and decoded observables]
Only the normalized rail populations $p_r^\pm$ must lie in $[0,1]$; raw $\delta m_r$ need not. The exactness claim is for the diagonal two-rail encoding and decoded expectations of rail number operators.
\end{remark}

\subsection{Practical interpretation}\label{subsec:practical}
At an operational level, Theorem~\ref{thm:main} and its surrounding lemmas can be read as a short list of circuit-level facts, independent of the specific lattice or moment basis in use.

\begin{enumerate}[leftmargin=2.2em,label=(I\arabic*)]
\item \textbf{One dissipative mode $=$ two rails.} Each $r\in\Dd$ is carried by its own two-qubit register $(p_r^+,p_r^-)$; no coupling between distinct modes is required.
\item \textbf{Dissipation is amplitude damping with survival $|\lambda_r|$.} The contractive factor $|\lambda_r|\le 1$ is realized physically as the survival probability of the $|1\rangle$ population on each rail under $\AD_{|\lambda_r|}$, not as a block-encoded amplitude rescaling requiring post-selection.
\item \textbf{Sign reversal is a rail SWAP, not a negative amplitude.} The overrelaxation branch $\lambda_r<0$, which is \emph{not} a free byproduct of any unitary acting on nonnegative populations, is captured by a classically conditioned SWAP of the two rails after damping. There is no ``negative amplitude,'' no coherent conditional on the sign of $\lambda_r$, and no additional heralding wiring.
\item \textbf{The channel stage is deterministic.} The map $\Ec_r$ in \eqref{eq:signed-channel} is CPTP by construction, so its per-application success probability is $1$ and the dissipative sweep does not accumulate any multiplicative success-probability penalty across modes, sites, or timesteps (Corollary~\ref{cor:deterministic}).
\item \textbf{Modes are mutually independent.} Since registers for distinct $r\in\Dd$ are disjoint, the full dissipative block is a tensor product of per-mode channels and can be executed in parallel over $\Dd$ without cross-talk (Remark~\ref{rem:composite}).
\item \textbf{The encoding scale $S_r$ is classical side information.} $S_r(\bm x,t)$ is a deterministic, classically tracked quantity used consistently at encoding and decoding, not a hidden quantum resource (Proposition~\ref{prop:adaptive}).
\item \textbf{Stencil- and basis-agnostic.} The construction assumes only modewise diagonal relaxation \eqref{eq:mrt-collision} in some invertible moment basis $M$ with $\lambda_r\in[-1,1]$; it does not depend on the particular velocity stencil (D2Q9, D3Q19, \ldots) or on orthogonality of $M$.
\end{enumerate}

\section{Block-encoding vs CPTP dissipation: the practical advantage}\label{sec:blockenc-vs-cptp}

The purpose of this section is to make explicit what the construction of Section~\ref{sec:encoding} delivers relative to the textbook route for realizing the contraction $\delta m_r\mapsto\lambda_r\delta m_r$ on a quantum computer. The essential comparison is not how the contraction is computed, but what price is paid per application in \emph{success probability}.

\subsection{The block-encoding route and its bottleneck}\label{subsec:be-bottleneck}
A standard way to implement a contractive, non-unitary linear map $A$ with $\|A\|\le 1$ on a quantum computer is to block-encode it inside a larger unitary $U$ acting on system and ancilla registers, so that a projection onto the $|0\rangle$-ancilla subspace yields $A$ (up to subnormalization $\alpha\ge \|A\|$) on the system~\cite{nielsen2010qcqi,watrous2018quantum}. For a modewise scalar contraction $\lambda_r$ with $|\lambda_r|\le 1$, the per-application success probability on a generic input state is bounded above by
\begin{equation}
 p_{\mathrm{succ}}^{(r)} \;\le\; \frac{|\lambda_r|^2}{\alpha_{r}^{2}} \;\le\; |\lambda_r|^{2},
 \label{eq:be-succ-mode}
\end{equation}
with equality achievable only in the tight block-encoding case $\alpha_r=1$. Composed independently across dissipative modes $r\in\Dd$, lattice sites $\bm x\in\Omega_h$, and timesteps $t=1,\ldots,T$, the end-to-end success probability for a block-encoded MRT dissipative sweep is bounded by
\begin{equation}
 p_{\mathrm{succ}}^{\mathrm{end}}
 \;\le\;
 \prod_{t=1}^{T}\prod_{\bm x\in\Omega_h}\prod_{r\in\Dd}
 |\lambda_r(\bm x,t)|^{2},
 \label{eq:be-succ-end}
\end{equation}
which tends to zero as either $|\Omega_h|\to\infty$ or $T\to\infty$ whenever some $|\lambda_r|<1$ (the generic case away from conservation and strict involutive limits). Amplitude amplification sharpens the prefactor but does not alter the asymptotic statement that \eqref{eq:be-succ-end} decays multiplicatively in $T$ and $|\Omega_h|$. More sophisticated constructions, including global block encodings that treat the diagonal operator $\Lambda$ over modes and sites jointly with a single subnormalization $\alpha\ge\max_r|\lambda_r|$, likewise improve the per-application prefactor by replacing the naive product bound with a state-dependent quantity $\|\Lambda\otimes I\,|\psi\rangle\|^{2}/\alpha^{2}$; they do not, however, remove the multiplicative decay over $T$ under repeated application when at least one dissipative mode is strictly contractive. This compounded probability is widely identified in the QLBM literature as the principal obstacle to end-to-end execution. Analogous multiplicative post-selection penalties, of the same structural form though not expressed in the mode-wise factors $|\lambda_r|$, arise in block-encoded Carleman-linearized LB pipelines~\cite{sanavio2024lbq}, whose depth/shot budgets are further driven by the non-sparse, non-block-diagonal structure of the block-encoded Carleman collision matrix~\cite{sanavio2024lbq}.

\paragraph{The same bottleneck under LCU.}
The linear-combination-of-unitaries (LCU) realization of a contraction $A=\sum_k c_k U_k$ with $c_k\ge 0$ and $\sum_k c_k=\|c\|_1$ prepares a ``select'' ancilla in $|c\rangle\propto\sum_k\sqrt{c_k}|k\rangle$, applies the multiplexed unitary $\sum_k |k\rangle\langle k|\otimes U_k$, and uncomputes the ancilla; projection of the ancilla on $|c\rangle$ extracts $A|\psi\rangle/\|c\|_1$, so the per-application success probability on a pure input is bounded by $\|A|\psi\rangle\|^{2}/\|c\|_1^{2}$, of the same form as \eqref{eq:be-succ-mode} with $\alpha_r$ replaced by $\|c\|_1$~\cite{nielsen2010qcqi,watrous2018quantum}. The linear-equilibrium QLBM of Liu, John, and Emerson~\cite{liu2026qlbmlinear} is a concrete realization of this route: they cast the Bhatnagar--Gross--Krook (BGK) collision (with bounce-back absorbed) as a non-unitary matrix $C$ in the velocity basis and implement it through singular value decomposition (SVD) followed by LCU of the resulting diagonal singular-value factor. The same bound therefore applies to their construction, and, a fortiori, to any mode-wise MRT extension in which $\Lambda=\mathrm{diag}(\lambda_r)$ is realized by LCU of its nonnegative-eigenvalue SVD factor: the product structure in \eqref{eq:be-succ-end} is inherited over $(T,|\Omega_h|,|\Dd|)$. The common cause is structural: both block-encoding and LCU target a \emph{coherent linear-operator} realization, in which the output is required to be the pure state $A|\psi\rangle$ up to normalization, and the lost norm necessarily reappears as failure probability on a measured ancilla.

\subsection{The deterministic CPTP route}\label{subsec:cptp-route}
The channel $\Ec_r$ defined in \eqref{eq:signed-channel} is completely positive and trace-preserving by construction: amplitude damping is CPTP~\cite{nielsen2010qcqi,watrous2018quantum}, tensor products of CPTP channels are CPTP, and composition with a unitary SWAP preserves CPTP. Consequently
\begin{equation}
 \Tr\!\big[\Ec_r(\rho_r)\big] \;=\; \Tr[\rho_r] \;=\; 1
 \qquad \text{for every valid two-rail input state } \rho_r,
\end{equation}
so the per-application success probability is $p_{\mathrm{succ}}^{(r)}=1$, independent of $\lambda_r$. The contraction on decoded populations arises from coherent probability transfer to the damping ancilla, which is subsequently discarded (or reset) rather than post-selected. The end-to-end success probability of a full MRT dissipative sweep is therefore
\begin{equation}
 p_{\mathrm{succ}}^{\mathrm{end}}
 \;=\;
 \prod_{t=1}^{T}\prod_{\bm x\in\Omega_h}\prod_{r\in\Dd} 1
 \;=\; 1,
\end{equation}
for every $T$ and every lattice size. The overrelaxation branch $\lambda_r<0$, which is sign-reversing and thus \emph{not} a free byproduct of block-encoded amplitude control, is absorbed into a \emph{classically} conditioned SWAP on the two rails at no probability cost.

\paragraph{Lindblad/GKSL viewpoint.}
At the density-matrix level, Theorem~\ref{thm:main} implements the dissipative collision as a completely positive trace-preserving channel built from amplitude damping, i.e.\ the standard open-system mechanism of Gorini--Kossakowski--Sudarshan--Lindblad (GKSL) theory, rather than as closed-system unitary (Schr\"odinger) evolution~\cite{nielsen2010qcqi,watrous2018quantum,manzano2020lindblad}. The continuous-time GKSL generator
\begin{equation}
 \dot\rho \;=\; -\mathrm{i}[H,\rho] \;+\; \sum_{k} \gamma_{k}\!\left(L_{k}\,\rho\,L_{k}^{\dagger} \;-\; \tfrac12\bigl\{L_{k}^{\dagger}L_{k},\,\rho\bigr\}\right),
 \qquad \gamma_{k}\ge 0,
 \label{eq:gksl}
\end{equation}
is the most general generator of a Markovian semigroup of CPTP maps on a finite-dimensional system~\cite{manzano2020lindblad}.
GKSL is normally written in continuous time ($\dot\rho=\mathcal{L}[\rho]$); here each LBM collision step applies one discrete CPTP map~$\Ec_r$, without identifying that step with a specific $e^{\Delta t\mathcal{L}}$ unless a lattice--time embedding is fixed separately.

In structural terms, the moment-basis transform $m=Mf$ plays the role of the diagonalization that decouples the dissipative GKSL generator into independent single-channel jump terms. The full dissipative collision factorizes over modes as $\Ec^{\mathrm{diss}}=\bigotimes_{r\in\Dd}\Ec_r$ (Remark~\ref{rem:composite}), with each $\Ec_r$ a single-channel amplitude-damping map on its two-rail encoding (Section~\ref{sec:encoding}). Under this analogy: the dissipative MRT index $r\in\Dd$ corresponds to the GKSL jump index $k$ in~\eqref{eq:gksl}; the MRT rate $s_r$ corresponds to the dissipation rate $\gamma_k$; the rail-lowering Kraus operator inside $\Ec_r$ corresponds to the jump operator $L_k$; and the conserved moments $(\rho,\bm j)$, which are exact fixed points of $\Ec^{\mathrm{diss}}$ because $\lambda_r=1$ on conserved $r$, correspond to the fixed-point manifold of the dissipative generator. If, in addition, a lattice--time embedding is fixed by identifying one MRT step with a continuous time interval $\Delta t$, then for $\lambda_r=1-s_r>0$ the per-step contraction $\alpha_r=|1-s_r|$ is the time-$\Delta t$ sample
\begin{equation}
 \alpha_r \;=\; e^{-\gamma_r\,\Delta t},
 \qquad
 \gamma_r \;=\; -\frac{1}{\Delta t}\,\ln\bigl|1-s_r\bigr|,
 \label{eq:lindblad-rate}
\end{equation}
of an exponential single-jump GKSL semigroup acting on the residual encoding. The overrelaxation branch $\lambda_r<0$ has no continuous-time semigroup analogue at the same single rate alone: it is realized in Theorem~\ref{thm:main} as that same amplitude-damping step composed with a unitary rail SWAP (Section~\ref{sec:encoding}), so it is a unitary dressing of a one-jump GKSL semigroup rather than an element of one.

In the same spirit as the standard Lindbladian jump terms in~\eqref{eq:gksl}, which dissipate coherences and drive $\rho$ toward an associated fixed manifold, MRT collision dissipates nonequilibrium kinetic information by relaxing moments toward $m^{\mathrm{eq}}$ (diagonally, $\delta m_r'=\lambda_r\delta m_r$ in \eqref{eq:mrt-collision}), or populations toward $f^{\mathrm{eq}}$ in a BGK-type picture $\Omega_{ij}(f_j-f_j^{\mathrm{eq}})$ after a change of basis; the analogy is at the level of \emph{relaxation to local equilibrium} and of a diagonalization-into-independent-channels structure, not an identity of linear algebraic form between superoperators on $\rho$ and the classical collision matrix on $(f,m)$.

\paragraph{Stinespring dilation, not coherent linear-operator realization.}
The structural reason the bound \eqref{eq:be-succ-mode} does not apply to the primitive proved here is that Section~\ref{sec:encoding} realizes a completely positive trace-preserving \emph{channel} on a rail density matrix, not a coherent linear operator on a pure rail state. By Stinespring's theorem~\cite{nielsen2010qcqi,watrous2018quantum}, every CPTP map admits the form
\begin{equation}
 \Ec_r(\rho_r) \;=\; \Tr_E\!\big[U_r\,(\rho_r\otimes |0\rangle\langle 0|_E)\,U_r^{\dagger}\big],
 \label{eq:stinespring}
\end{equation}
for some unitary $U_r$ on system plus damping environment $E$, in which the environment is \emph{discarded} (or reset) rather than measured and post-selected. Trace preservation is therefore automatic, and the per-application success probability is $1$ by construction (Corollary~\ref{cor:deterministic}). The bound that governs block-encoding \eqref{eq:be-succ-mode} and LCU is a theorem about the coherent recovery of $A|\psi\rangle$ from a pure input via a projectively measured ancilla; it does not apply to the Stinespring/discard realization solved here, in which the decoded output \eqref{eq:decoder} is a linear functional of the rail density matrix on diagonal observables. The contraction $p_r^\pm\mapsto\alpha_r p_r^\pm$ induced by amplitude damping on rail populations is, at the level of coherent amplitudes, a non-unitary map whose lost norm is coherently \emph{transferred} to the damping ancilla and subsequently reset; what is saved relative to block-encoding and LCU is probability loss, not the entropy/Landauer cost of the reset itself, which is charged to the environment in the standard open-system way.

\subsection{What is traded in exchange}\label{subsec:tradeoffs}
The deterministic CPTP realization is not free: it trades post-selection for a different, classically accountable, set of costs.
\begin{enumerate}[leftmargin=2.2em,label=(T\arabic*)]
\item \textbf{A nonlinear classical encoding step.} The map $\delta m_r\mapsto(p_r^+,p_r^-)$ in \eqref{eq:rails} is nonlinear because of the positive/negative split (Remark~\ref{rem:nonlinear-encoding}). In a hybrid loop where $\delta m_r$ is classically known at each use, this split is free; in a fully coherent pipeline, a reversible sign-split gadget is required and must be accounted for separately.
\item \textbf{Two rails per dissipative mode.} One extra qubit per mode relative to a hypothetical single-rail encoding. This is a \emph{qubit} overhead, not a probability overhead.
\item \textbf{Classically indexed scale.} The encoding scale $S_r(\bm x,t)$ is classical side information used consistently in encoding and decoding (Proposition~\ref{prop:adaptive}); it does not introduce coherent ancilla cost, but it must be tracked.
\end{enumerate}
In every QLBM pipeline known to the authors, converting per-step probability loss \eqref{eq:be-succ-end} into the classical bookkeeping (T1)--(T3) is the strictly dominant trade.

\subsection{Summary comparison}\label{subsec:compare-table}

\begin{table}[htbp]
\centering
\caption{Methodological comparison of the two possible routes for realizing the diagonal MRT dissipative contraction $\delta m_r\mapsto \lambda_r\delta m_r$ on a quantum computer: the standard block-encoding route with ancilla post-selection (Section~\ref{subsec:be-bottleneck}), and the deterministic two-rail CPTP route introduced here (Theorem~\ref{thm:main}, Section~\ref{subsec:cptp-route}). The comparison is between the two \emph{approaches}, not between specific published works; entries on the right-hand column refer to the CPTP primitive of this paper. Quantities are per dissipative mode unless stated otherwise.}
\label{tab:be-vs-cptp}
\small
\begin{tabular}{lcc}
\toprule
Quantity & Block-encoding route (post-selected) & Two-rail CPTP route (this paper) \\
\midrule
Per-step success probability & $\le |\lambda_r|^2/\alpha_r^2$ & $1$ \\
End-to-end over $(T,|\Omega_h|,|\Dd|)$ & decays multiplicatively & $1$ \\
Overrelaxation $\lambda_r<0$ & extra conditional structure & classically conditioned SWAP \\
Ancilla per mode & $\ge 1$ flag, post-selected & $1$ damping ancilla, reset \\
Post-selection / heralding & required & none \\
Amplitude amplification & typically required & not applicable \\
Qubits carrying $\delta m_r$ & $1$ rail & $2$ rails (sign split) \\
Encoding of $\delta m_r$ & linear & nonlinear (sign split; Rem.~\ref{rem:nonlinear-encoding}) \\
Scale $S_r$ handling & implicit in normalization & explicit classical side info \\
\bottomrule
\end{tabular}
\end{table}

\subsection{Relation to recent QLBM work}\label{subsec:literature}
Representative algorithmic lines in the recent literature include nonlinear fluid QLBM~\cite{wang2025npj}, linearized advection--diffusion lattice Boltzmann schemes with dynamic circuits and related linear-collision QLBM formulations~\cite{wawrzyniak2025cpc,xu2025pre,zeng2025pof}, linear-equilibrium constructions that treat collision and streaming linear-algebraically and decompose the collision via SVD--LCU (singular value decomposition and linear combination of unitaries)~\cite{liu2026qlbmlinear}, multi-circuit resource reduction~\cite{lee2026multiple}, one-step device-oriented lattice Boltzmann algorithms~\cite{bastida2026qlbmdevices}, lattice Boltzmann--Carleman quantum algorithms and circuits for moderate Reynolds number flows~\cite{sanavio2024lbq}, learned surrogate circuits for dissipative BGK collisions~\cite{lacatus2026surrogate}, and a recent synthesis of realizability issues~\cite{tiwari2025algorithmic}. With the exception of learned surrogates, all of the above face, to varying degrees, multiplicative post-selection or heralding penalties in their dissipative/contractive block: the mode-wise form \eqref{eq:be-succ-end} applies directly to block-encoded MRT schemes, and structurally analogous decays arise in LCU-based realizations of linear collision operators~\cite{liu2026qlbmlinear} and in block-encoded Carleman-linearized collision matrices~\cite{sanavio2024lbq} (with the caveat, in the Carleman case, that the non-unitarity is carried by a non-diagonal operator in velocity/Carleman-variable space rather than by a moment-basis diagonal $\Lambda$). Theorem~\ref{thm:main} provides a block-encoding-free alternative for the mode-diagonal case: it isolates the dissipative MRT block as a deterministic CPTP primitive, leaving equilibrium loading, streaming, and boundaries to be composed on top, as discussed in Section~\ref{sec:integration}.

\section{Integration in an MRT collision pipeline}\label{sec:integration}

Theorem~\ref{thm:main} establishes a block-encoding-free CPTP realization of the dissipative MRT step. For that result to be useful, the primitive must compose inside a recognizable circuit pipeline. This section describes, at a schematic level, how the primitive slots into a \emph{generic block-encoded MRT collision pipeline} at the site level, and comments on its relation to Carleman-linearized QLBM circuits~\cite{sanavio2024lbq}, which share the same block-encoding bottleneck but do not expose a mode-wise diagonal target that the primitive could substitute for directly. The description is deliberately generic: the host pipeline is any quantum-compatible realization of the standard MRT collision stages on a site, and the primitive replaces only the contractive dissipative block within that host. Full coherent-state preparation, streaming, boundaries, and readout are not treated here.

\subsection{A generic MRT collision pipeline}\label{subsec:generic-pipeline}
At the site level, any quantum-compatible realization of one MRT collision decomposes into four logically distinct stages, composed per timestep:

\begin{enumerate}[leftmargin=2.2em,label=(P\arabic*)]
\item \textbf{Equilibrium preparation.} Compute or encode the equilibrium $m^{\mathrm{eq}}(\rho,\bm u)$ in an arithmetic register. The specific realization of this stage, including whether division is performed explicitly or postponed by a numerator-register representation, is part of the host pipeline and is not fixed here.
\item \textbf{Nonequilibrium extraction.} Form $\delta m_r$ for each dissipative mode $r\in\Dd$ on the arithmetic register of the host pipeline.
\item \textbf{Dissipative CPTP block (Theorem~\ref{thm:main}).} For each $r\in\Dd$, prepare the signed two-rail register $\rho_r=\rho_r^+\otimes\rho_r^-$ via \eqref{eq:rails}, apply $\Ec_r$ in \eqref{eq:signed-channel}, and decode \eqref{eq:decoder} to obtain $\lambda_r\delta m_r$ on the arithmetic register. The classical sign of $\lambda_r$ selects the SWAP branch at compile time; no coherent conditional on $\lambda_r$ is required.
\item \textbf{Reconstruction.} Combine the relaxed $\delta m'$ with $m^{\mathrm{eq}}$ and apply $M^{-1}$ to obtain post-collision populations on the host's output register.
\end{enumerate}
In compact notation, stages (P2)--(P3) realize the vector-form dissipative update
\begin{equation}
\delta m' \;=\; \Lambda\,\delta m,
\qquad
\Lambda \;=\; \operatorname{diag}\!\Big((1)_{r\in\Cc},\,(\lambda_r)_{r\in\Dd}\Big),
\label{eq:vector-lambda}
\end{equation}
on the moment register, where $\Lambda$ is modewise diagonal by construction of MRT collision in the moment basis~$M$ and the conserved block of $\Lambda$ is the identity. At each site $\bm x\in\Omega_h$, stage (P3) acts on the dissipative rails as the node-level composite channel
\begin{equation}
\Ec_{\bm x}^{\mathrm{diss}} \;\eqdef\; \bigotimes_{r\in\Dd}\Ec_{r,\bm x},
\label{eq:diss-node}
\end{equation}
with each $\Ec_{r,\bm x}$ given by \eqref{eq:signed-channel}; $\Ec_{\bm x}^{\mathrm{diss}}$ is CPTP as a tensor product of CPTP maps (Corollary~\ref{cor:deterministic}) and reproduces the diagonal action of $\Lambda$ restricted to $\Dd$ on the decoded moments (Theorem~\ref{thm:main}, Remark~\ref{rem:composite}). Because registers for different dissipative modes are disjoint, stage (P3) is parallel across $r\in\Dd$. Because (P3) is deterministic, \emph{no success-probability discount accumulates across modes, sites, or timesteps}: the probabilistic cost of the whole dissipative sweep is $p_{\mathrm{succ}}=1$, matching Section~\ref{subsec:cptp-route}.

\subsection{Resource sketch per dissipative mode}\label{subsec:resources}
At the level of abstract circuit primitives, the dissipative CPTP block (P3) for one mode $r$ requires
\begin{itemize}[leftmargin=2em]
\item two rail qubits (the signed two-rail register);
\item one damping ancilla per rail (reset after use), realizing $\AD_{\alpha_r}$ with survival $\alpha_r=|\lambda_r|$ via the standard two-qubit damping gadget (one controlled rotation of angle $2\arccos\sqrt{\alpha_r}$ from rail to ancilla, one CNOT, one reset)~\cite{nielsen2010qcqi,watrous2018quantum};
\item one rail SWAP applied iff $\lambda_r<0$ (classically conditioned, compile-time);
\item no heralding flag and no post-selection control wiring.
\end{itemize}
The angle $2\arccos\sqrt{\alpha_r}$ is classical, fixed by the lattice scheme, and can be computed once per mode. Across all $|\Dd|$ dissipative modes and $|\Omega_h|$ sites, the per-timestep dissipative-block cost is linear in $|\Dd|\,|\Omega_h|$, with no probability discount. A detailed gate-count breakdown for stages (P1), (P2), and (P4) depends on the specific host pipeline and is outside the scope of this paper.

\subsection{Scope of drop-in replacement and relation to other QLBM pipelines}\label{subsec:dropin}
The dissipative CPTP block (P3) acts only on the signed two-rail register at a fixed site and mode, given a classical target relaxation multiplier $\lambda_r$. It is therefore pipeline-agnostic \emph{within the MRT family}: any host circuit whose dissipative collision is implemented as a block-encoded mode-wise diagonal contraction $\Lambda=\mathrm{diag}(\lambda_r)$ on nonequilibrium moments (equivalently, any block-encoded MRT scheme that carries $\delta m_r$ as an accessible register after a moment-basis transform) can replace its block-encoded dissipative block with Theorem~\ref{thm:main}, upgrading the contribution of that block to the end-to-end success probability from~\eqref{eq:be-succ-end} to~$1$. The cost paid in exchange is (T1)--(T3) of Section~\ref{subsec:tradeoffs}, dominated by the classical nonlinear sign split; this split is free whenever $\mathrm{sign}(\delta m_r)$ is classically available from the host's moment register.

Constructions in which the non-unitarity is \emph{not} of mode-wise diagonal form in a moment basis lie outside this drop-in statement. The Carleman-linearized Lattice Boltzmann circuit of Sanavio and Succi~\cite{sanavio2024lbq} is representative: it linearizes the BGK collision by Carleman tensor-embedding rather than by MRT mode decomposition, and its non-unitary core is the full (non-sparse, non-block-diagonal) Carleman collision matrix, not a diagonal $\Lambda$ on nonequilibrium moments. The present primitive does not substitute for that non-unitarity as such, and the Carleman truncation does not, on its own, expose MRT nonequilibrium moments $\delta m_r$ with classical sign structure; making the primitive applicable in that setting would require, additionally, an explicit moment-basis transform of the Carleman state. It is nevertheless relevant to Carleman-LB as a \emph{target for reformulation}: any Carleman-LB variant that exposes an explicit MRT mode decomposition of its relaxation block would gain the same replacement, at no probability cost for the dissipative block. A minimal such reformulation is sketched in Section~\ref{subsec:carleman-cptp} below.

\subsection{A hybrid Carleman--CPTP pipeline}\label{subsec:carleman-cptp}
The observation that motivates this subsection is that the utility of Carleman linearization in QLBM is \emph{the coherent evaluation of a nonlinear equilibrium} $f^{\mathrm{eq}}(\rho,\bm u)$, which is polynomial (typically quadratic) in~$f$; the use of Carleman for the dissipative BGK relaxation itself is a structural consequence of coupling equilibrium and relaxation in a single non-unitary operator, not an intrinsic requirement. Separating the two responsibilities, with Carleman used only for equilibrium and the MRT-diagonal open channel used for relaxation, yields a pipeline in which the present primitive applies without modification.
\paragraph{Host-agnostic composite-map decomposition.}
The combination of open-channel CPTP dissipation with the mainstream unitary machinery of QLBM is \emph{not specific to Carleman linearization}: the same schema covers Carleman linearization, linear combination of unitaries (LCU), and block-encoding (BE) interchangeably, since all three are unitary (or block-encoded unitary) realizations of \emph{linear} operators and can therefore occupy the same pre/post slots. Abstracting away this choice, a single MRT collision step at a site can be written at the \emph{composite-map} level as a sandwich of three logically distinct stages,
\begin{equation}
\boxed{\;\;
\Phi_{\mathrm{step}}
\;=\;
\underbrace{\mathcal{U}_{\mathrm{post}}}_{\substack{\text{unitary mainstream} \\ \text{(reconstruction,} \\ M^{-1}\text{, streaming, \ldots)}}}
\;\circ\;
\underbrace{\Ec^{\mathrm{diss}}}_{\substack{\text{open-channel CPTP} \\ \text{(this paper,} \\ p_{\mathrm{succ}}=1\text{)}}}
\;\circ\;
\underbrace{\mathcal{U}_{\mathrm{pre}}}_{\substack{\text{unitary mainstream} \\ \text{(Carleman / LCU /} \\ \text{block-encoding,} \\ \text{moment extract, \ldots)}}}
\;\;}
\label{eq:high-level-hybrid}
\end{equation}
\noindent where $\mathcal{U}_{\mathrm{pre}}$ is any coherent (up to fixed subnormalizations absorbed in the host encoding) preparation that delivers the nonequilibrium moments $\delta m$ on an accessible register, $\Ec^{\mathrm{diss}}=\bigotimes_{r\in\Dd}\Ec_r$ is the deterministic CPTP dissipation of \eqref{eq:signed-channel} (site-level; cf.\ \eqref{eq:diss-node}), and $\mathcal{U}_{\mathrm{post}}$ is any coherent recombination and inverse moment transform that returns post-collision populations. Equation~\eqref{eq:high-level-hybrid} is the general combination schema: \emph{the mainstream unitary machinery is confined to $\mathcal{U}_{\mathrm{pre}}$ and $\mathcal{U}_{\mathrm{post}}$, and the only intrinsically non-unitary stage is $\Ec^{\mathrm{diss}}$, which is realized deterministically with unit per-step success probability}. Any dissipation-dependent multiplicative success-probability decay of the form \eqref{eq:be-succ-end} is therefore removed at the level of the collision step itself, regardless of which unitary mainstream method is chosen for $\mathcal{U}_{\mathrm{pre}}$ and $\mathcal{U}_{\mathrm{post}}$. Concretely, \eqref{eq:high-level-hybrid} admits at least three interchangeable specializations, labeled (U1)--(U3) in the list below:
\begin{enumerate}[leftmargin=2.2em,label=(U\arabic*)]
\item \textbf{Block-encoding (BE) host.} $\mathcal{U}_{\mathrm{pre}},\mathcal{U}_{\mathrm{post}}$ are block-encoded realizations of the linear moment transform $M$, equilibrium subtraction, and inverse moment transform $M^{-1}$; the present primitive $\Ec^{\mathrm{diss}}$ replaces the block-encoded, post-selected realization of the diagonal contraction $\Lambda=\operatorname{diag}(\lambda_r)$ on dissipative modes (Sections~\ref{subsec:generic-pipeline} and~\ref{subsec:dropin}).
\item \textbf{LCU host.} $\mathcal{U}_{\mathrm{pre}},\mathcal{U}_{\mathrm{post}}$ implement the linear stages through a linear combination of unitaries (LCU), with the non-unitary contractive factor handled separately. In extant LCU-based QLBM constructions the contractive factor is itself realized by LCU of a non-unitary collision matrix, either directly or via SVD of the collision: Liu, John, and Emerson~\cite{liu2026qlbmlinear} apply SVD followed by LCU to the non-unitary BGK collision in the velocity basis; the present primitive $\Ec^{\mathrm{diss}}$ replaces that LCU-of-a-contraction step, eliminating the associated penalty $p_{\mathrm{succ}}\le\|A|\psi\rangle\|^{2}/\|c\|_1^{2}$ (notation of Section~\ref{subsec:be-bottleneck}) on the dissipative stage.
\item \textbf{Carleman host (hybrid Carleman--CPTP).} $\mathcal{U}_{\mathrm{pre}}$ additionally contains the Carleman lift and coherent evaluation of the nonlinear equilibrium, $f^{\mathrm{eq}}=E_K\tilde f$, detailed in \eqref{eq:hybrid-stages} below; $\Ec^{\mathrm{diss}}$ replaces the dissipative part of the fused Carleman-BGK block~\cite{sanavio2024lbq}.
\end{enumerate}
Here BE stands for \emph{block-encoding}, the standard embedding of a non-unitary map inside a larger unitary with ancilla post-selection (Section~\ref{subsec:be-bottleneck}), and LCU for \emph{linear combination of unitaries} (Section~\ref{subsec:be-bottleneck}, paragraph on LCU); both are the mainstream coherent techniques for the linear stages of QLBM. The quantum-algebraic realization of $\Ec^{\mathrm{diss}}$ itself (amplitude-damping gadget, Stinespring dilation, rail SWAP under $\lambda_r<0$) is the subject of Section~\ref{sec:encoding} and is not repeated here; \eqref{eq:high-level-hybrid} is stated at the composite-map level on purpose, so that the combination claim reads the same whatever the host's unitary stack happens to be.
\paragraph{Carleman equilibrium lift.}
The remainder of this subsection works out, as a concrete example, the Carleman-host specialization (U3) of the general schema~\eqref{eq:high-level-hybrid}: we fix $\mathcal{U}_{\mathrm{pre}}$ and $\mathcal{U}_{\mathrm{post}}$ to the Carleman lift plus linear moment stages detailed in \eqref{eq:hybrid-stages} below, keep $\Ec^{\mathrm{diss}}$ as the open-channel CPTP block of Section~\ref{sec:encoding}, and trace through the resulting composite map and probability budget. The parallel specializations (U1) and (U2) follow the same template with the Carleman lift dropped (no nonlinear equilibrium to evaluate). Specialization (U1) (BE host) coincides with the generic pipeline of Sections~\ref{subsec:generic-pipeline} and \ref{subsec:dropin}; specialization (U2) (LCU host) inherits the same pre/post schema, with the LCU success-probability bound of Section~\ref{subsec:be-bottleneck} (paragraph ``The same bottleneck under LCU'') replaced by $p_{\mathrm{succ}}=1$ on the dissipative stage. Both are algebraically simpler than the Carleman case and are therefore not written out explicitly here.

Let $K\ge 2$ denote the Carleman truncation order. For each site $\bm x$, introduce the truncated Carleman tower
\begin{equation}
\tilde f \;=\; f \,\oplus\, (f\otimes f) \,\oplus\,\cdots\,\oplus\, f^{\otimes K},
\label{eq:carleman-tower}
\end{equation}
amplitude-encoded on a register $|\tilde f\rangle$ of the host pipeline. Because $f^{\mathrm{eq}}(\rho,\bm u)$ is polynomial of degree $\le K$ in~$f$, there exists a linear operator $E_K$ on the Carleman tower such that
\begin{equation}
f^{\mathrm{eq}} \;=\; E_K\,\tilde f
\label{eq:carleman-equilibrium}
\end{equation}
exactly up to the standard Carleman truncation error~\cite{sanavio2024lbq}. Realizing $E_K$ as a block-encoded (or LCU) unitary is the conventional Carleman-LB task and is host-specific; the present subsection does not fix that realization.

\paragraph{Pipeline composition.}
Fixing an invertible moment basis $M$ as in Section~\ref{subsec:generic-mrt}, with conserved/dissipative index split $(\Cc,\Dd)$ and modewise multipliers $\lambda_r$ as in \eqref{eq:lambda-def}, the hybrid site-level update decomposes into the following stages, replacing the single fused Carleman-BGK block with a composition of linear stages and a mode-wise CPTP block:
\begin{equation}
\begin{aligned}
&\text{(H1)}\quad |f\rangle \;\longrightarrow\; |\tilde f\rangle
&&\text{(Carleman lift to order $K$)},\\
&\text{(H2)}\quad m \;=\; M f,\qquad m^{\mathrm{eq}} \;=\; M\,E_K\,\tilde f
&&\text{(linear; moment-basis transform)},\\
&\text{(H3)}\quad \delta m \;=\; m - m^{\mathrm{eq}}
&&\text{(linear; on the moment register)},\\
&\text{(H4)}\quad \delta m'_r \;=\; \lambda_r\,\delta m_r
&&\text{(Theorem~\ref{thm:main}; CPTP open channel, } p_{\mathrm{succ}}=1\text{)},\\
&\text{(H5)}\quad m^+ \;=\; m^{\mathrm{eq}} + \delta m'
&&\text{(linear; reconstruction)},\\
&\text{(H6)}\quad f^\star \;=\; M^{-1}\,m^+
&&\text{(linear; inverse moment transform).}
\end{aligned}
\label{eq:hybrid-stages}
\end{equation}
Equivalently, at the composite-map level,
\begin{equation}
f^\star
\;=\;
M^{-1}\!\bigl(\,M E_K\,\tilde f \;+\; \Lambda\,(M f - M E_K\,\tilde f)\,\bigr),
\qquad
\Lambda = \operatorname{diag}\!\bigl((1)_{r\in\Cc},(\lambda_r)_{r\in\Dd}\bigr),
\label{eq:hybrid-composite}
\end{equation}
in which the action of $\Lambda$ on the dissipative modes is realized by the CPTP channel $\bigotimes_{r\in\Dd}\Ec_r$ of \eqref{eq:signed-channel} on signed two-rail registers (Theorem~\ref{thm:main}, Remark~\ref{rem:composite}). The resulting register-level pipeline, with (H1)--(H6) slotted into the host-agnostic decomposition \eqref{eq:high-level-hybrid}, is summarized in Figure~\ref{fig:pipeline-carleman}.

\begin{figure}[htbp]
\centering
\includegraphics[width=\textwidth]{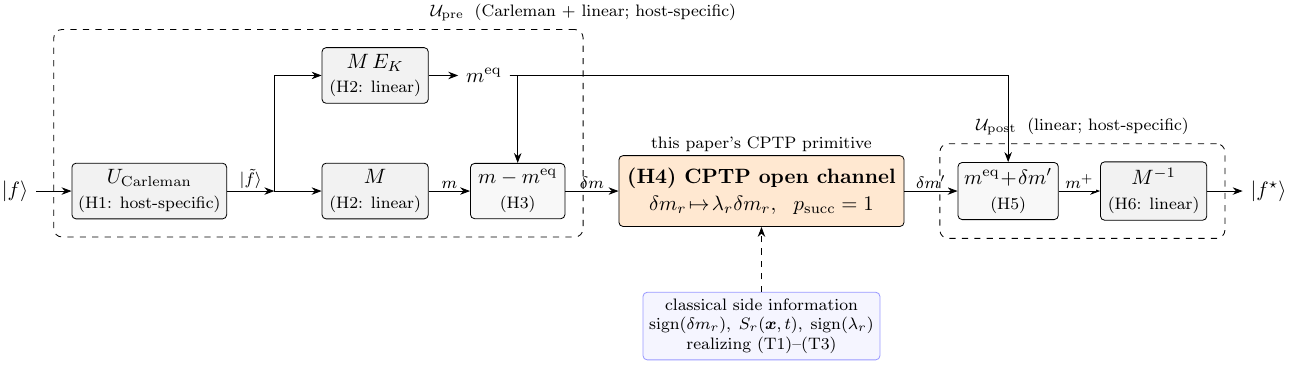}
\caption{Register-level block diagram of the hybrid Carleman--CPTP collision step \eqref{eq:hybrid-stages}, realizing specialization (U3) of the host-agnostic schema \eqref{eq:high-level-hybrid}. The outer dashed frames mark the induced grouping $\mathcal{U}_{\mathrm{pre}}=$(H1)$\circ$(H2)$\circ$(H3) and $\mathcal{U}_{\mathrm{post}}=$(H5)$\circ$(H6). The orange block (H4) is the CPTP open channel of Theorem~\ref{thm:main}, whose signed two-rail internal structure is detailed in Figures~\ref{fig:schematic}--\ref{fig:circuit-h4} (block and gate level). The dashed arrow entering (H4) from below carries the classical side information $\bigl(\mathrm{sign}(\delta m_r),\,S_r(\bm x,t),\,\mathrm{sign}(\lambda_r)\bigr)$ realizing the bookkeeping items (T1)--(T3) of Section~\ref{subsec:tradeoffs} (encoding-scale choice: Remark~\ref{rem:choose-Sr}; adaptive scaling: Proposition~\ref{prop:adaptive}).}
\label{fig:pipeline-carleman}
\end{figure}

\paragraph{Probability budget.}
Under the hybrid decomposition \eqref{eq:hybrid-stages}:
\begin{itemize}[leftmargin=2em,topsep=0pt,itemsep=1pt,partopsep=0pt]
\item Stages (H2), (H3), (H5), (H6) are \emph{linear} and unitary up to fixed subnormalizations absorbed in the host encoding of $M$, $M^{-1}$, and arithmetic registers; they carry no \emph{dissipation-dependent} success-probability loss.
\item Stage (H4) is the only intrinsically non-unitary step in the dissipative collision. By Corollary~\ref{cor:deterministic} it is \emph{deterministic}, so
\begin{equation}
p_{\mathrm{succ}}^{\mathrm{diss,\,end}}
\;=\;
\prod_{t=1}^{T}\prod_{\bm x\in\Omega_h}\prod_{r\in\Dd} 1 \;=\; 1,
\label{eq:hybrid-psucc}
\end{equation}
removing the mode-wise penalty \eqref{eq:be-succ-end} that the fused Carleman-BGK block incurs.
\item Stage (H1) is the standard Carleman lift. Its realization cost (depth, subnormalization, truncation order) is host-determined and shared with any Carleman-based pipeline~\cite{sanavio2024lbq}; it is \emph{not} affected by the choice of dissipation primitive.
\end{itemize}
The hybrid pipeline therefore converts the success-probability penalty of the Carleman-BGK collision into the classical bookkeeping (T1)--(T3) of Section~\ref{subsec:tradeoffs} on the relaxation block, while leaving the Carleman machinery strictly as a coherent equilibrium evaluator.

\paragraph{Scope of the claim.}
The present paper proves (H4). Stages (H1)--(H3) and (H5)--(H6) are linear and standard in the LBM and Carleman-LB literature; a full specification and resource accounting of (H1) (Carleman truncation order, block-encoding of $E_K$, memory scaling in the tower \eqref{eq:carleman-tower}) and (H2), (H6) (coherent moment transforms) are host-design choices and are left to future work. The claim here is architectural: the hybrid pipeline \eqref{eq:hybrid-stages} is a concrete, minimal example of a Carleman-based QLBM construction in which the dissipative relaxation is handled by the present CPTP primitive with unit per-step success probability, illustrating the ``target for reformulation'' statement at the end of Section~\ref{subsec:dropin}.

\subsection{Scope of this section}\label{subsec:integration-scope}
The claim of this section is that the dissipative block (P3) is host-agnostic and does not contribute a success-probability penalty to the broader pipeline. Quantitative resource estimates for (P1), (P2), and (P4), coherent preparation of the rails directly from a population register, multi-step composition with streaming and boundaries, and finite-shot readout are properties of the surrounding host, not of the primitive itself, and are therefore determined by the specific architecture in which the primitive is embedded.

\section{Numerical consistency check}\label{sec:validation}
Because Theorem~\ref{thm:main} is algebraic, we verify the dissipative open-channel construction numerically in two complementary ways. The D3Q19 audit in Section~\ref{subsec:d3q19-audit} uses nonequilibrium moment fields from a \emph{full} collide--stream evolution (decaying Taylor--Green vortex) on the D3Q19 stencil, so the residuals in Table~\ref{tab:theorem-audit} and Figure~\ref{fig:theorem-audit} are probed on \emph{realistic, flow-generated} $\delta m_r$ data rather than on hand-chosen or purely random values alone. The stencil-free sweeps in Section~\ref{subsec:synthetic-audit} apply the \emph{same} encoder \eqref{eq:rails}, channel \eqref{eq:signed-channel}, and decoder \eqref{eq:decoder} to inputs $(\delta m_r,\lambda_r)\in[-X,X]\times[-1,1]$ drawn with no connection to a lattice, thereby stress-testing the per-mode map independently of a particular velocity set or flow. Together, these two checks cover both a complete classical LBM time history in a standard 3D setting and a purely algebraic, stencil-agnostic stress test.

\subsection{D3Q19 lattice audit (Taylor--Green vortex)}\label{subsec:d3q19-audit}
We check the open-channel construction on the D3Q19 stencil (Appendix~\ref{app:d3q19}) by evaluating the maximum over dissipative modes and lattice sites of $|\delta m_r'-\lambda_r\delta m_r|$ along the open-channel path. All tests use the \emph{same} decaying Taylor--Green vortex run on the periodic box $[0,L)^3$ with $L=2\pi$, $N^3=64^3$, $\tau=0.5035$, $u_0=0.1$, and adaptive per-node rail scales $S_r(\bm x,t)=\max(|\delta m_r|,\varepsilon)$; Table~\ref{tab:theorem-audit} summarizes five audits over the short window $t\in[0,\,19.6]$, while Figure~\ref{fig:theorem-audit} shows the same error series continued over $t\in[0,\,981.7]$. The reported time is the advective LBM time $t = i\,(2\pi/N)\,u_0$ after $i$ collide--stream iterations (per-iteration increment $\approx9.82\times10^{-3}$), equivalently $t^\ast\equiv t/L$ box-crossings at~$u_0$.

Five audits are reported. (i)~Standard $\lambda_r\in[-1,1]$ from the D3Q19 $s_r$ at $\tau=0.5035$, running maximum over the trajectory. (ii--iv)~Endpoint regression with $\lambda_r$ uniformly set to $-1$, $0$, $+1$ on all dissipative modes, applied to a representative Taylor--Green vortex (TGV) pre-collision snapshot, probing maximal overrelaxation (SWAP with $\alpha_r=1$), full dissipation ($\alpha_r=0$), and the identity branch. (v)~Adaptive $S_r$ coincides with~(i) on this run. In all cases the decoded output matches $\lambda_r\delta m_r$ to well within IEEE-754 double-precision roundoff.

\begin{table}[htbp]
\centering
\caption{Maximum error $\max_{r\in\Dd,\,\text{sites}}|\delta m_r'-\lambda_r\delta m_r|$ for the five audits of Theorem~\ref{thm:main} on D3Q19, in IEEE-754 double precision; run defined in Section~\ref{sec:validation}. Endpoint tests set all dissipative $\lambda_r$ to the same value and act on a TGV pre-collision snapshot.}
\label{tab:theorem-audit}
\begin{tabular}{lcc}
\toprule
Test & Quantity audited & Max error \\
\midrule
TGV trajectory, standard $\lambda_r$ & $|\delta m_r'-\lambda_r\delta m_r|$ & $4.44\times10^{-16}$ \\
Endpoint $\lambda_r=-1$ (SWAP, $\alpha=1$) & $|\delta m_r'+\delta m_r|$ & $6.16\times10^{-33}$ \\
Endpoint $\lambda_r=0$ ($\alpha=0$) & $|\delta m_r'|$ & $0$ \\
Endpoint $\lambda_r=+1$ (identity) & $|\delta m_r'-\delta m_r|$ & $6.16\times10^{-33}$ \\
TGV trajectory, adaptive $S_r(\bm x,t)$ & $|\delta m_r'-\lambda_r\delta m_r|$ & $4.44\times10^{-16}$ \\
\bottomrule
\end{tabular}
\par\medskip
\noindent
\begin{minipage}{\linewidth}
\footnotesize
\emph{Note.} For $\lambda_r=\pm 1$ we have $\alpha_r=1$ (damping is the identity; a rail \textsc{swap} is applied only if $\lambda_r=-1$). The $6.16\times10^{-33}$ entries are \emph{absolute} maxima, not $O(1)$ relative accuracies: they are set by a few sites with extremely small~$|\delta m_r|$, for which the floating-point Kraus implementation leaves a tiny residual of order $\mathcal{O}(\varepsilon_{\mathrm{mach}}\,|\delta m_r|)$ (with $S_r=\max(|\delta m_r|,\varepsilon)$; Remark~\ref{rem:choose-Sr}). The $4.44\times10^{-16}$ trajectory rows are at the expected IEEE-754 double-precision scale for a full D3Q19 time history.
\end{minipage}
\end{table}

Figure~\ref{fig:theorem-audit} resolves the standard-$\lambda_r$ error per iteration over the extended window, confirming that the machine-precision agreement of Table~\ref{tab:theorem-audit} holds along the dynamics rather than at isolated snapshots.

\begin{figure}[htbp]
\centering
\includegraphics[width=0.68\textwidth]{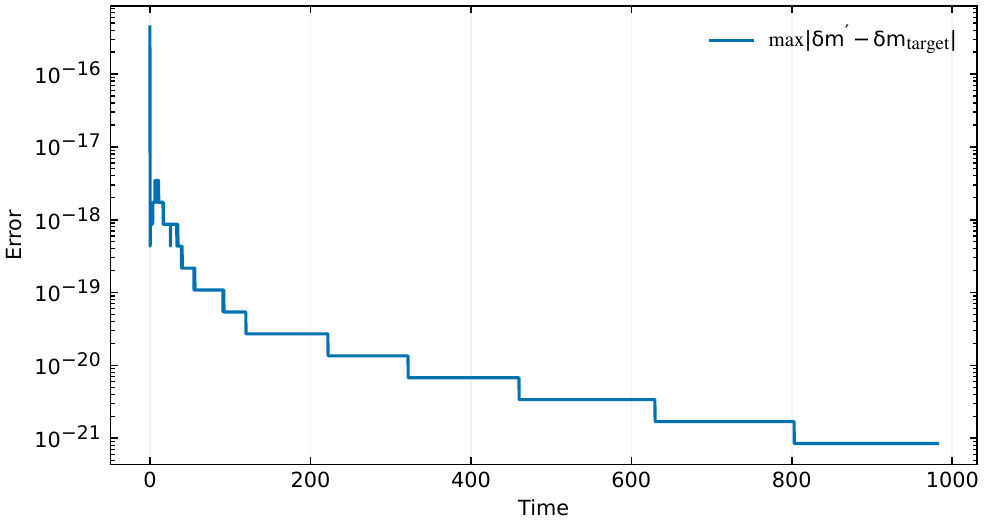}
\caption{Dissipative-mode audit $\max_{r\in\Dd,\,\text{sites}}|\delta m_r'-\lambda_r\delta m_r|$ versus advective LBM time for the run defined in Section~\ref{subsec:d3q19-audit}. The error stays at the IEEE-754 double-precision floor throughout.}
\label{fig:theorem-audit}
\end{figure}

\subsection{Stencil-free synthetic sweep}\label{subsec:synthetic-audit}
To confirm that the exactness of Theorem~\ref{thm:main} is an algebraic property of the construction, not an artifact of a particular stencil or flow, we repeat the audit on inputs that are \emph{not} drawn from any lattice simulation. Random pairs $(\delta m_r,\lambda_r)\in[-X,X]\times[-1,1]$ (with $X=1$) are fed through the same signed two-rail encoder \eqref{eq:rails}, channel \eqref{eq:signed-channel}, and decoder \eqref{eq:decoder} used in Section~\ref{subsec:d3q19-audit}, with $S_r=\max(|\delta m_r|,\varepsilon)$ in the adaptive mode, and the pointwise residual $|\delta m_r'-\lambda_r\delta m_r|$ is aggregated over samples.

Five sweeps are reported in Table~\ref{tab:synthetic-audit}:
\begin{enumerate}[leftmargin=2.2em,label=(S\arabic*)]
\item \textbf{Dense $\lambda$ grid.} $101$ uniform $\lambda$ values in $[-1,1]$ with $1000$ random $\delta m$ per $\lambda$ ($1.01\times10^{5}$ samples), covering the admissible range at resolution $2\times10^{-2}$ in $\lambda$.
\item \textbf{Uniform random $(\delta m,\lambda)$.} $200\times200=4\times10^{4}$ samples drawn jointly uniformly, probing generic interior points.
\item \textbf{Joint boundary stress.} $\lambda\in\{-1,\,-1\!+\!\varepsilon_{\mathrm{b}},\,-\varepsilon_{\mathrm{b}},\,0,\,+\varepsilon_{\mathrm{b}},\,1\!-\!\varepsilon_{\mathrm{b}},\,1\}$ with $\varepsilon_{\mathrm{b}}=10^{-12}$; for each such $\lambda$, $10^{4}$ random interior $\delta m\in[-X,X]$ are combined with the deterministic $\delta m$-edge vector $\{-X,\,-X\!+\!\varepsilon_{\mathrm{b}},\,-\varepsilon_{\mathrm{b}},\,0,\,+\varepsilon_{\mathrm{b}},\,X\!-\!\varepsilon_{\mathrm{b}},\,X\}$, so the joint corners $(\lambda\text{-edge},\,\delta m\text{-edge})$ are probed deterministically rather than only probabilistically. This regime includes exact overrelaxation ($\lambda=-1$, $\alpha_r=1$, SWAP), exact identity ($\lambda=+1$), and exact full dissipation ($\lambda=0$).
\item \textbf{Scale-sensitivity.} $200$ random $(\delta m,\lambda)$ pairs, each evaluated at $50$ log-uniform scale values $S_r\in[\max(|\delta m|,\varepsilon_S),\,10^{6}\max(|\delta m|,\varepsilon_S)]$ with $\varepsilon_S=10^{-12}$, verifying that the decoded output continues to match $\lambda\,\delta m$ to roundoff at every admissible $S_r$, consistent with Proposition~\ref{prop:adaptive}.
\item \textbf{Deterministic exact-corner audit.} The $5\times 5$ grid $(\delta m,\lambda)\in\{-X,-X/2,0,X/2,X\}\times\{-1,-1/2,0,1/2,1\}$, reader-verifiable by hand; this complements the random sweeps (S1)--(S4) with a fully reproducible grid of corner and half-corner points.
\end{enumerate}

\begin{table}[htbp]
\centering
\caption{Stencil-free synthetic audit of Theorem~\ref{thm:main}: maximum residual $|\delta m_r'-\lambda_r\delta m_r|$ over random and deterministic draws on $[-X,X]\times[-1,1]$ with $X=1$, in IEEE-754 double precision.}
\label{tab:synthetic-audit}
\begin{tabular}{lrc}
\toprule
Sweep & \#samples & Max error \\
\midrule
(S1) Dense $\lambda$ grid, random $\delta m$ & $101{,}000$ & $1.11\times10^{-16}$ \\
(S2) Uniform random $(\delta m,\lambda)$ & $40{,}000$ & $1.11\times10^{-16}$ \\
(S3) Joint boundary stress ($\lambda$- and $\delta m$-edges) & $70{,}049$ & $1.11\times10^{-16}$ \\
(S4) Scale sensitivity over $S_r$ & $10{,}000$ & $3.33\times10^{-16}$ \\
(S5) Deterministic exact corners $5\times 5$ & $25$ & $1.11\times10^{-16}$ \\
\bottomrule
\end{tabular}
\end{table}

All five sweeps saturate at the IEEE-754 double-precision floor (machine $\varepsilon\approx 2.22\times10^{-16}$; the $3.33\times10^{-16}$ entry for (S4) corresponds to at most $1.5\,\varepsilon$, consistent with a single division by $S_r$ at decode). In particular, (S3) exhibits no degradation at the exact boundaries $\lambda\in\{-1,0,+1\}$ even when combined with $|\delta m|=X$, where the channel reduces to a permutation, a reset, and the identity respectively; (S4) confirms that varying $S_r$ over six decades above $|\delta m|$ leaves the decoded output matching $\lambda\,\delta m$ at the roundoff floor, as required by Proposition~\ref{prop:adaptive}; and (S5) returns machine-precision agreement at every point of the deterministic $5\times 5$ corner grid.

\section{Conclusion}
We proved that the signed diagonal MRT dissipative update
\[
\delta m_r'=\lambda_r\delta m_r,\qquad \lambda_r\in[-1,1],
\]
admits an exact decoded-expectation realization by a \emph{deterministic} local CPTP channel on signed two-rail registers: amplitude damping with survival $|\lambda_r|$ on each rail, followed by a rail SWAP in the overrelaxation branch $\lambda_r<0$ (Theorem~\ref{thm:main}). Because the channel is trace-preserving by construction, the per-step success probability of the dissipative block is unity and does not compound multiplicatively across modes, sites, or timesteps (Corollary~\ref{cor:deterministic}), in contrast to block-encoded realizations of the same contraction whose end-to-end probability is bounded by \eqref{eq:be-succ-end}. Adaptive scale choices $S_r(\bm x,t)$ are classically indexed side information rather than hidden quantum resources (Proposition~\ref{prop:adaptive}); within the present nonnegative population-encoding framework a single rail cannot exactly represent a signed moment (Theorem~\ref{thm:one-rail}, Appendix~\ref{app:one-rail}), which motivates the two-rail split. Numerical audits, combining a D3Q19 decaying Taylor--Green vortex trajectory covering the standard $\lambda_r$ used in that run, the endpoints $\lambda_r\in\{-1,0,+1\}$, and adaptive scales (Table~\ref{tab:theorem-audit}, Figure~\ref{fig:theorem-audit}) with a stencil-free synthetic sweep over random and deterministic $(\delta m_r,\lambda_r)\in[-X,X]\times[-1,1]$, including dense $\lambda$ grids, uniform random pairs, joint boundary stress over both $\lambda$- and $\delta m$-edges, scale-sensitivity tests, and a deterministic $5\times 5$ exact-corner audit (Table~\ref{tab:synthetic-audit}), confirm machine-precision agreement with the target dissipative update across both the physical D3Q19 use case and the general algebraic domain.

This paper isolates the dissipative block of MRT collision as a block-encoding-free CPTP primitive and shows, schematically, how it slots into a generic block-encoded MRT collision pipeline (Section~\ref{sec:integration}). Block-encoded Carleman-linearized LB circuits~\cite{sanavio2024lbq} share the same structural block-encoding bottleneck but do not directly admit the present primitive as a drop-in replacement, since their non-unitarity is not of mode-wise diagonal form in a moment basis; any Carleman-LB variant that exposes an explicit MRT mode decomposition of its relaxation block would, however, inherit the replacement, and a minimal hybrid pipeline of this kind, with Carleman linearization used strictly for the nonlinear equilibrium and open-channel CPTP used for the diagonal dissipative relaxation, is sketched in Section~\ref{subsec:carleman-cptp}. Composition of the primitive into a full one-node D3Q19 MRT collision, together with exact multi-step lattice evolution under streaming and standard computational fluid dynamics (CFD) boundaries, is left to future work; Appendix~\ref{app:classical-embed} records only the algebraic consequence of embedding the primitive in a classical MRT baseline that shares the same moment basis and equilibrium.

\appendix
\phantomsection
\section*{Appendices}
\addcontentsline{toc}{section}{Appendices}
\section{D3Q19 stencil for the numerical audit}\label{app:d3q19}
The numerical consistency check in Section~\ref{sec:validation} uses standard athermal D3Q19 MRT in the d'Humi\`eres/Lallemand--Luo ordering \cite{dhumieres2002multiple,lallemand2003theory} ($q=19$, second-order $f^{\mathrm{eq}}$, $m^{\mathrm{eq}}=Mf^{\mathrm{eq}}$, collision and streaming \eqref{eq:mrt-collision}--\eqref{eq:streaming}). Conserved indices are $\Cc=\{0,3,5,7\}$ \cite{dhumieres2002multiple}; all other modes are dissipative. The tabulated relaxation parameters $s_r$ are the usual D3Q19 entries used together with this ordering (as in \cite{dhumieres2002multiple,lallemand2003theory} and in the simulation that produces Figure~\ref{fig:theorem-audit}); they are not fixed by the theorem.

\section{One-rail no-go within the present population-encoding framework}\label{app:one-rail}
This appendix records the obstruction that motivates the two-rail split in Section~\ref{sec:encoding}, \emph{within the population-encoding framework of this paper}: diagonal one-qubit population states and population-expectation decoders of the form $\rho\mapsto\sigma S\,\Tr[n\rho]$. It does not concern other one-qubit encodings (e.g.\ centered affine decoders on Bloch coordinates), which can realize a bounded signed scalar on one qubit.

\begin{theorem}[No exact signed one-rail population encoding]\label{thm:one-rail}
Fix $X>0$ and $S>0$. Consider a one-rail population representation of a scalar $x$ by diagonal one-qubit states
\begin{equation}
 \rho^{(1)}(x)=(1-p(x))\,|0\rangle\langle 0|+p(x)\,|1\rangle\langle 1|,
 \qquad p(x)\in[0,1],
\end{equation}
together with a single-rail population decoder of the form
\begin{equation}
 \Dc^{(1)}_{S,\sigma}(\rho)\eqdef \sigma\,S\,\Tr[n\rho],
 \qquad
 n=|1\rangle\langle 1|,
 \qquad
 \sigma\in\{+1,-1\}.
\end{equation}
Then there is no choice of $p:[-X,X]\to[0,1]$ such that
\begin{equation}
 \Dc^{(1)}_{S,\sigma}\!\big(\rho^{(1)}(x)\big)=x
 \qquad \text{for all } x\in[-X,X].
\end{equation}
Consequently, within this one-rail nonnegative population framework, a single rail cannot exactly encode a signed scalar on any domain containing both positive and negative values. In particular, the signed MRT dissipative update $\delta m_r'=\lambda_r\delta m_r$ on such a domain cannot be realized exactly by a single nonnegative population rail with a population-expectation decoder.
\end{theorem}

\begin{proof}
If $\sigma=+1$, then for every state of the stated form,
\[
\Dc^{(1)}_{S,+1}(\rho)=S\,\Tr[n\rho]=S\,p(x)\in[0,S].
\]
Hence the decoded values are nonnegative and cannot equal all $x\in[-X,X]$, which include negative values. If $\sigma=-1$, then $\Dc^{(1)}_{S,-1}(\rho)=-S\,p(x)\in[-S,0]$, so the decoded values are nonpositive and cannot equal all $x\in[-X,X]$, which include positive values. Therefore no such one-rail population encoding with the stated decoder can exactly represent a signed scalar on a domain containing both signs.
\end{proof}

\begin{remark}[Why sign reversal is the real obstruction]
Within a single nonnegative population rail, decoded values have a fixed sign by construction. The obstruction is therefore not the branch $\lambda_r<0$ \emph{per se}, but the attempt to represent a signed scalar by a single nonnegative population expectation. The two-rail split of Section~\ref{sec:encoding} resolves this by storing positive and negative parts separately and letting the decoder take their difference.
\end{remark}

\section{Classical embedding: populations and one-step consequence}\label{app:classical-embed}
This appendix is not part of the main theorem: it records standard algebraic consequences if the dissipative update of Theorem~\ref{thm:main} is combined with the same moment reconstruction and streaming as the classical MRT baseline.

Theorem~\ref{thm:main} gives $\delta m_r'=\lambda_r\delta m_r$ for each $r\in\Dd$. For conserved modes, define $\delta m_r'=\delta m_r$. The full nodewise update is therefore
\begin{equation}
 \delta m_r'=
 \begin{cases}
 \delta m_r, & r\in\Cc,\\
 \lambda_r\delta m_r, & r\in\Dd.
 \end{cases}
 \label{eq:full-delta-prime}
\end{equation}
Set
\begin{equation}
 m^+ = m^{\mathrm{eq}} + \delta m',
 \qquad
 f^\star_{\mathrm{open}} = M^{-1}m^+.
 \label{eq:open-post}
\end{equation}

\begin{assumption}[Shared baseline data]\label{ass:shared}
The open-system construction and the classical MRT baseline use the same moment basis $M$ and the same equilibrium map $m^{\mathrm{eq}}(\rho,\bm u)$.
\end{assumption}

\begin{proposition}[Matching post-collision populations under shared baseline]\label{prop:pop-match}
Assume Assumption~\ref{ass:shared}. When the channel of Theorem~\ref{thm:main} is used as the dissipative module with scales $S_r=S_r(\bm x,t)$ chosen so that $p_r^\pm\in[0,1]$ at each node and with the same $S_r$ used in encoding and decoding, the vector $f^\star_{\mathrm{open}}$ defined in \eqref{eq:open-post} equals the classical post-collision vector $f^\star$ in \eqref{eq:mrt-collision}.
\end{proposition}
\begin{proof}
For conserved modes, \eqref{eq:full-delta-prime} is the identity update. For dissipative modes, Theorem~\ref{thm:main} gives $\delta m_r'=\lambda_r\delta m_r=(1-s_r)\delta m_r$. Hence $\delta m'$ coincides with the classical MRT-updated nonequilibrium moment vector. Substituting into \eqref{eq:open-post} and \eqref{eq:mrt-collision} yields the same $m^+$ and therefore the same $f^\star=M^{-1}m^+$.
\end{proof}

\begin{remark}[One-step map and Chapman--Enskog]\label{rem:ce-appendix}
Under Assumption~\ref{ass:shared}, if decoded post-collision populations agree with classical MRT at each node (Proposition~\ref{prop:pop-match}) and streaming is the classical permutation \eqref{eq:streaming}, then the deterministic one-step update agrees with the corresponding classical MRT-LBM scheme, and standard Chapman--Enskog analysis for that classical scheme applies unchanged. This observation is included only to connect the dissipative primitive to the classical discrete dynamics.
\end{remark}

\end{document}